\ifdefined\pdfminorversion
\fi
\documentclass[preprint,12pt]{elsarticle}

\usepackage[T1]{fontenc}
\usepackage[utf8]{inputenc}
\usepackage{lmodern}
\usepackage{amsmath,amssymb,amsfonts,amsthm,mathtools}
\usepackage{bm}
\usepackage{graphicx}
\usepackage{xcolor}
\usepackage{microtype}
\usepackage{booktabs}
\usepackage{tabularx}
\usepackage{array}
\usepackage{multirow}
\usepackage{siunitx}
\usepackage{enumitem}
\usepackage{algorithm}
\usepackage{float}
\usepackage[noend]{algpseudocode}
\usepackage{xspace}
\usepackage{tikz}
\usepackage{pgfplots}
\usepackage{hyperref}
\usepackage[nameinlink,capitalise,noabbrev]{cleveref}

\graphicspath{{figures/}}
\hypersetup{
  colorlinks=true,
  linkcolor=black,
  citecolor=black,
  urlcolor=black,
  pdfauthor={Arman Ferdowsi, Maryam DehghanChenary, Kevin Tierney, Atakan Aral},
  pdftitle={A Fault-Tolerant Spike-Time Interface for Approximate Agreement in Distributed Neuromorphic Systems}
}
\pgfplotsset{compat=1.18}
\usetikzlibrary{arrows.meta,positioning,calc,fit,shapes.geometric,decorations.pathreplacing}

\newtheorem{theorem}{Theorem}[section]
\newtheorem{lemma}[theorem]{Lemma}

\newtheorem{corollary}[theorem]{Corollary}
\theoremstyle{definition}
\newtheorem{definition}[theorem]{Definition}
\newtheorem{example}[theorem]{Example}
\newtheorem{observation}[theorem]{Observation}
\theoremstyle{remark}
\newtheorem{remark}[theorem]{Remark}

\newcommand{\eps}{\varepsilon}
\newcommand{\clip}{\operatorname{clip}}

\newcommand{\Dec}{\operatorname{Dec}}
\newcommand{\SpikeTrim}{\textnormal{\textsc{SpikeTrim}}\xspace}
\newcommand{\SIF}{\textnormal{\textsc{SIF}}\xspace}
\newcommand{\FirstSpike}{\textnormal{\textsc{FirstSpike}}\xspace}
\newcommand{\Byz}{\mathcal B}
\newcommand{\Correct}{\mathcal C}
\newcommand{\Nodes}{\mathcal V}

\crefname{algorithm}{Algorithm}{Algorithms}
\Crefname{algorithm}{Algorithm}{Algorithms}
\crefname{observation}{observation}{observations}
\Crefname{observation}{Observation}{Observations}

\setlist[itemize]{leftmargin=*,topsep=4pt,itemsep=2pt}
\setlist[enumerate]{leftmargin=*,topsep=4pt,itemsep=2pt}
\journal{Neurocomputing}
\biboptions{sort&compress}

\begin{document}

\makeatletter
\def\ps@pprintTitle{%
  \let\@oddhead\@empty
  \let\@evenhead\@empty
  \let\@oddfoot\@empty
  \let\@evenfoot\@empty
}
\makeatother

\begin{frontmatter}

\title{A Fault-Tolerant Spike-Time Interface for Approximate Agreement in Distributed Neuromorphic Systems}

\author[cs]{Arman Ferdowsi\corref{cor1}}
\ead{arman.ferdowsi@univie.ac.at}

\author[bda]{Maryam DehghanChenary}
\ead{maryam.dehghan.chenary@univie.ac.at}

\author[bda]{Kevin Tierney}
\ead{kevin.tierney@univie.ac.at}

\author[cs]{Atakan Aral}
\ead{atakan.aral@univie.ac.at}

\cortext[cor1]{Corresponding author}

\address[cs]{Faculty of Computer Science, University of Vienna, Vienna, Austria}

\address[bda]{Faculty of Business, Economics and Statistics, University of Vienna, Vienna, Austria}

\begin{abstract}
Large neuromorphic systems contain many processing tiles that may replicate a shared control parameter such as a threshold reference. If these copies diverge, identical inputs may be processed under different intended settings. We study how tiles can reduce this disagreement when communication carries only labeled spike times and up to \(f\) sender labels may be Byzantine. A raw event stream cannot supply the one-value-per-sender input required by classical approximate agreement because a faulty sender can remain silent, flood a receiver, or report different times to different receivers.

We introduce the Spike-time Interface for Faults, or \SIF, which combines paced epochs, sender attribution, per-label \FirstSpike admission, bounded timing error, and a silence sentinel. For an affine one-spike code, midpoint decoding attains the exact deterministic minimax error \(\rho=\min\{1/2,\omega/L\}\), where \(\omega\) is the residual timing uncertainty and \(L\) is the usable encoding window. \SpikeTrim applies the classical mean-subsequence-reduced (MSR) rule to the sender-indexed decoded values. For \(n\ge3f+1\), it guarantees one-step robust validity, the tight noiseless contraction factor \(f/(n-2f)\) under direct updates, an explicit worst-case asymptotic disagreement bound, and finite recovery after transient agreement-state corruption. A closed-form test determines whether a validated timing budget meets a target disagreement. Simulations illustrate the fault threshold, timing dependence, flooding resistance, and recovery. A controlled spiking classifier experiment shows an association between faster control-state alignment and lower prediction disagreement under a finite maintenance budget.
\end{abstract}


\end{frontmatter}

\section{Introduction}
\label{sec:introduction}
Modern neuromorphic processors distribute neurons, synapses, memory, and routing across processing units commonly called \emph{cores} or \emph{tiles}. In this paper, each participating tile is modeled as exactly one sender-labeled node. Communication is event-based and commonly carries source information through an address-event or network-on-chip fabric~\cite{Boahen2000AER,Davies2018Loihi,Merolla2014TrueNorth,Purohit2022AER,Balaji2023NeuSB,Kudithipudi2025Scale}.
Mixed-signal implementations exhibit mismatch in firing thresholds, biases,
weights, and time constants. Their parameters may also vary with temperature,
time, and circuit aging~\cite{Buechel2021Robust,Pehle2022BrainScaleS,
Song2021Reliability}. Local device calibration compensates for mismatch and
drift in the mapping from a logical reference to tile-specific circuit settings. However, it does not ensure that different tiles hold the same logical reference. This gap creates a system-level inconsistency risk. We therefore consider systems in which tiles maintain replicas of a shared slow
reference, such as a threshold target, gain target, or homeostatic setpoint. If those replicas diverge, tiles implement mutually incompatible parameter regimes, causing the same logical input to be processed under different logical settings across the system even when every tile is locally
calibrated.

We study a low-rate control-plane maintenance primitive, separate from
application-level inference, that reduces disagreement among replicas of one
slow scalar. For this approximate-agreement task, robust validity keeps each
correct update within the current correct-state interval up to bounded decoding
error, while convergence drives the range among correct replicas toward a
timing-limited neighborhood. These guarantees concern mutual consistency rather
than absolute anchoring. With persistent timing error, all correct replicas may
drift together away from the original input range.

During each maintenance epoch, every correct tile encodes its current scalar
replica in one sender-labeled spike time. If reliable scalar messages are
available, classical digital approximate agreement is simpler. We instead address the
constrained setting in which spike timing is the only real-valued communicated
observation. Sender labels and epoch identifiers remain required metadata.
Section~\ref{sec:task-level} later examines whether improved reference alignment
is associated with greater prediction consistency across tiles.

Reducing replica disagreement also has a task-level interpretation. Let
\(h_i(x)\) denote the scalar response of correct tile \(i\) when its local
reference equals \(x\). Assume that every \(h_i\) is \(K\)-Lipschitz on
\([0,1]\) and that \(|h_i(x)-h_j(x)|\le\varepsilon_h\) for all correct
\(i,j\) and all \(x\in[0,1]\). Then, for any correct tiles \(i\) and \(j\),
\[
\begin{aligned}
  |h_i(x_i)-h_j(x_j)|
  &\le
  |h_i(x_i)-h_i(x_j)|
  +
  |h_i(x_j)-h_j(x_j)| \\
  &\le
  K|x_i-x_j|+\varepsilon_h
  \le KR+\varepsilon_h,
\end{aligned}
\]
where \(R\) is the range of the correct reference states. Thus, reducing \(R\)
controls only the contribution caused by reference divergence.
\(\varepsilon_h\) captures residual same-reference mismatch after calibration.
This bound applies to continuous responses. Discrete class decisions need not
be Lipschitz, so the classifier study later reports an empirical association
rather than a consequence of this bound.
\begin{example}[Replicated threshold reference]
\label{ex:threshold}
Consider seven tiles that hold a normalized threshold reference \(x_i\in[0,1]\). At most two tiles can be faulty. Once per maintenance epoch, every correct tile broadcasts one labeled spike whose time encodes its current reference. The goal is for the five or more correct tiles to approach a common reference without allowing two faulty labels to pull a correct update outside the current correct range, apart from the unavoidable decoding error. The integer threshold \(n\ge 3f+1\) is met exactly in this example.
\end{example}

We use Byzantine behavior as a conservative model of what a faulty sender label can make a receiver observe. A faulty label may remain silent, emit many events, or present different timings to different receivers. This does not claim that physical faults are strategic. It ensures that the proof covers crash, stuck-at, flooding, duplicate-emission, and timestamp-corruption behaviors whenever they are confined to at most \(f\) labels. Timing perturbations on correct labels are bounded separately by \(\eta\). The classical deterministic resilience threshold under this fault model is \(n\ge3f+1\). The trimmed-mean rule described next removes the \(f\) smallest
and \(f\) largest sender-indexed values. Narrower fault models may require less
redundancy.

Classical synchronous Byzantine approximate agreement~\cite{Dolev1986ApproxAgree} provides the aggregation objective, but not the spike-time interface. Dolev et al. remove extreme values and average a selected subsequence of the retained values. The later Mean-Subsequence-Reduced (MSR) family includes the all-retained uniform-mean rule used here~\cite{Kieckhafer1994MixedMode}. Each receiver sorts one value attributed to every sender, removes the \(f\) smallest and \(f\) largest values, and averages the remaining \(n-2f\). However, a raw spike-event stream does not by itself provide the sender-indexed scalar input required by MSR. Correct observations are perturbed by delivery delay, timestamp jitter, quantization, and residual clock skew. A faulty sender can remain silent, emit many events, or send receiver-dependent event times. The multiplicity problem already breaks validity in the smallest nontrivial instance.

\begin{observation}[Raw events are not sender values]
\label{obs:raw-events}
Let \(n=4\), \(f=1\), and let all three correct values be zero. Suppose the faulty label emits \(K\ge 2\) spikes that all decode to one. A raw-event implementation forms the multiset with three zeros and \(K\) ones. After removing one low and one high event, its mean is
\[
  \frac{K-1}{K+1}>0.
\]
The update leaves the singleton correct input range even though only one sender is faulty. If the receiver first keeps at most one event from each sender label, the effective multiset is \(\{0,0,0,1\}\). MSR then returns zero.
\end{observation}

This example shows why a refractory or de-duplication rule must be indexed by sender. A global refractory period would suppress legitimate events from other senders. Trimming raw events does not repair the problem because the fault budget counts faulty labels, not emitted events.

We introduce the Spike-time Interface for Faults, or \SIF, as a receiver-side adapter between the event channel and MSR. It requires paced epochs, sender labels that cannot be forged at the immediate receiver, bounded timing uncertainty, and per-label \FirstSpike admission. The receiver accepts only the earliest event from each label in an epoch and uses a local sentinel for silence. The result is exactly one bounded value per sender label, receiver, and epoch. The classical MSR trimmed mean is then applied to the \SIF output.

The timing part requires more care than simply dividing an arrival timestamp by the epoch length. A correct receiver subtracts a nominal link delay and obtains a normalized observation of the form
\[
  s=Lx+z,
  \qquad L=T-2g,
  \qquad |z|\le \omega.
\]
Here \(x\in[0,1]\) is the encoded scalar, \(z\) is the residual timing error, and \(L\) is the usable encoding length. The epoch duration is \(T\), \(g\) is the guard at each endpoint, and \(\omega\) bounds \(|z|\). We use the midpoint of all values that could have produced \(s\). Its exact deterministic minimax error is
\[
  \rho=\min\!\left\{\frac12,\frac{\omega}{L}\right\}.
\]
We also prove a matching lower bound for deterministic decoders of this single affine observation. The result is not a lower bound against alternative encoders, multi-epoch coding, or multiple spikes.

Figure~\ref{fig:overview} summarizes the resulting pipeline. The interface first turns a noisy and potentially flooded event stream into a sender-indexed bounded-error multiset. \SpikeTrim then applies MSR and a convex state update.

\begin{figure}[t]
  \centering
  \includegraphics[width=\linewidth]{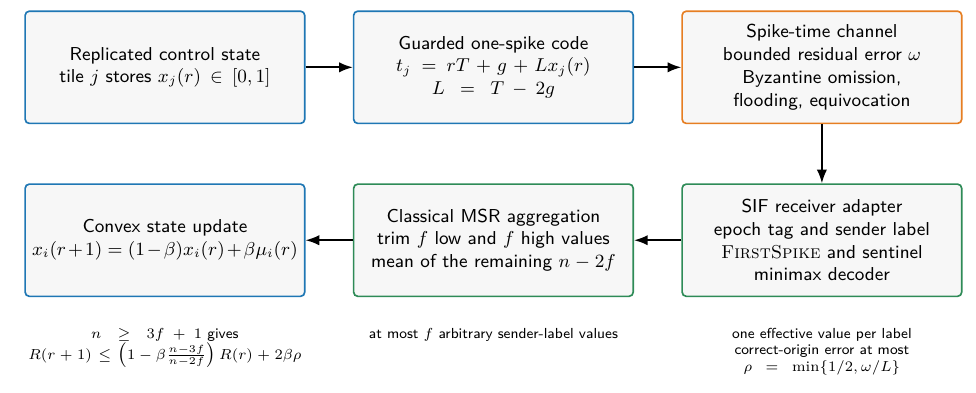}
  \caption{Receiver-centric pipeline. One guarded spike time carries the scalar. \SIF supplies paced epochs, immediate-receiver sender attribution, per-label \FirstSpike admission, a sentinel for silence, and minimax decoding. It produces exactly one bounded value per sender label, which is the input required by the classical MSR rule.}
  \label{fig:overview}
\end{figure}

\bigskip
\noindent
The paper makes four main contributions.

\begin{enumerate}
\item \textbf{A spike-time feasible-set bound.} We derive the feasible interval induced by guarded one-spike timing. Specializing classical Chebyshev-center estimation to this interval gives the exact radius \(\rho=\min\{1/2,\omega/L\}\). A channel-specific indistinguishability construction shows that this value is exact for the affine timing observation under the stated uncertainty class.

\item \textbf{An end-to-end design criterion.} We propagate the decoder radius through the all-retained MSR update. For the direct update with \(\beta=1\), we recover its classical noiseless factor \(f/(n-2f)\)~\cite{Kieckhafer1994MixedMode}. We then derive the bounded-error recurrence with spike-interface-specific constants, the resulting asymptotic range bound, and a closed-form timing-budget criterion.

\item \textbf{A sender-indexed spike-time interface.} We define \SIF as the adapter that turns a noisy, silent, flooded, or equivocating event stream into exactly one bounded value per sender label, receiver, and epoch. The transfer theorem makes the composition explicit. The raw-event counterexample shows why admission must be indexed by sender.

\item \textbf{Model-based evidence.} Simulations test contraction, timing uncertainty, the \(3f+1\) boundary, sender flooding, recovery, and finite-budget prediction consistency. The evaluation is a mechanism study rather than hardware validation or a general neural benchmark.
\end{enumerate}

\noindent
The scalar minimax rule and the reduced-sequence mean have classical antecedents. Our contribution is the spike-time feasible set, the sender-indexed event interface, and the explicit end-to-end constants that link timing uncertainty to Byzantine range and timing-budget guarantees.

The practical promise of the interface is quantitative. A designer can obtain and validate a deterministic upper bound on residual timing uncertainty, choose an epoch budget, and test whether the proved \SpikeTrim bound meets a target. Passing the test establishes that the proved model-relative bound meets the target under the stated timing, attribution, visibility, and fault assumptions. Failing it means that this analysis cannot establish the target. It does not prove that every one-spike protocol must fail. Tighter per-link nominal-delay calibration, additional correct-node redundancy,
a longer timing window, multiple events, or a digital payload can tighten the
guarantee.

\medskip
\noindent
\emph{Paper organization.} Section~\ref{sec:related} positions the work. Section~\ref{sec:model} defines \SIF and the agreement task. Section~\ref{sec:decoder} derives the one-spike decoder. Section~\ref{sec:protocol} presents \SpikeTrim. Section~\ref{sec:analysis} proves correctness and convergence. Section~\ref{sec:evaluation} reports the experiments. Section~\ref{sec:discussion} discusses implementation boundaries and limitations.

\section{Related work}
\label{sec:related}

\subsection{Byzantine approximate agreement}

Scalar approximate agreement asks correct processes to remain valid with respect to correct values while reducing their disagreement. In synchronous complete networks with receiver-dependent Byzantine messages, deterministic approximate agreement is possible exactly when \(n\ge 3f+1\)~\cite{Dolev1986ApproxAgree}. MSR trimming is a classical robust-aggregation mechanism for this setting. Later work extends iterative Byzantine consensus to directed and sparse graphs through graph robustness and related conditions~\cite{Vaidya2012Iterative,LeBlanc2013Robust}.

Persistent bounded measurement errors are known to replace exact consensus by a nonzero worst-case disagreement bound in otherwise fault-free linear consensus~\cite{Garulli2011BoundedErrors}. Here the perturbation radius is derived from spike timing and combined with receiver-dependent Byzantine MSR.

Recent systems use approximate agreement as a practical primitive for target localization and distributed oracles~\cite{Bandarupalli2024SensorBFT,Bandarupalli2024Delphi}. Work on asynchronous approximate agreement studies communication cost and progress when digital messages arrive without paced rounds~\cite{MizrahiErbesWattenhofer2026AA}. These settings assume explicit digital values or symbols. Our problem is different. A receiver observes one noisy time coordinate, and a faulty sender can create an unbounded raw event multiplicity unless the interface limits admission by label.

Self-stabilization concerns recovery from arbitrary transient state corruption without an external reset~\cite{Dijkstra1974SelfStab}. Byzantine clock synchronization shows that paced timing services can themselves be constructed under explicit network and clock assumptions~\cite{LenzenRybicki2019ClockSync}. Neuromorphic architectures also treat synchronization across cores as an explicit service and may replace global barriers with local synchronization~\cite{Li2025NeuroScale}. We do not construct such a service. We use a bounded-skew epoch service as a lower-layer obligation and prove conditional recovery of the agreement layer once the service again satisfies its stated assumptions.

\subsection{Spiking computation and temporal codes}

Temporal coding represents information through spike times. Time-to-first-spike models and learning rules study how neural systems encode or infer values from those times~\cite{Bohte2002SpikeProp,Bonilla2022TTFS}. Distributed-computing models of spiking networks study circuit size, memory, and communication complexity after a spike semantics has been fixed~\cite{LynchMuscoParter2017NeuroRAM,HitronMuscoParter2020SNNStreaming,HitronParterPerri2020AsyncNeural}.

Agreement has also been implemented as a computational task inside a spiking
neural network (SNN)~\cite{kunev2022agreement}. That problem concerns agreement computed by a neural network rather than a Byzantine communication interface for replicated real-valued state. Resilient pulse-coupled oscillator protocols adapt MSR-style filtering to malicious pulse timing. They include detection of multiple pulses and a pure-pulse protocol for phase and frequency synchronization~\cite{iori2024resilient,YanIshii2025SecurePCO}. These are the closest prior pulse interfaces. Their states are oscillator phase and frequency, and the pulses are part of the oscillator dynamics. \SIF instead encodes an arbitrary bounded scalar in one sender-labeled noisy time coordinate and derives the decoder error presented to MSR.

Our question precedes the neural computation performed above the channel. We ask which event-interface properties are needed for a receiver to obtain one robust scalar observation per sender. The key obstacles are sender multiplicity, sender attribution, epoch comparability, silence, and timing uncertainty. This focus is complementary to temporal-code accuracy and to the computational complexity of SNNs.

\subsection{Neuromorphic systems and reliability}

Large neuromorphic processors distribute computation across cores and event routers~\cite{Davies2018Loihi,Merolla2014TrueNorth,Modha2023NorthPole,Frenkel2023Design,Kudithipudi2025Scale}. Address-event communication provides a natural place to carry a source label~\cite{Boahen2000AER,Purohit2022AER}. Calibration and robust training address device-specific mismatch, while reliability studies address memories, neuron faults, aging, routing, permanent faults, and dynamic fault management~\cite{Buechel2021Robust,Pehle2022BrainScaleS,Song2021Reliability,Putra2021Respawn,Putra2022SoftSNN,Putra2023RescueSNN,Spyrou2021NeuronFault,Balaji2023NeuSB}.

Those mechanisms protect components or communication resources but do not reconcile replicas of a shared control reference under Byzantine faults. \SIF instead assumes that lower layers supply bounded timing and source attribution, then adds sender-indexed endpoint admission and conservatively treats each compromised or malfunctioning sender label as Byzantine.

\section{The spike-time interface and agreement task}
\label{sec:model}

\subsection{Nodes, faults, and logical visibility}
Let \(\Nodes=\{1,\ldots,n\}\) be a set of sender-labeled nodes. At most \(f\) nodes are Byzantine. For the convergence and recovery analysis, the Byzantine set \(\Byz\) is fixed, \(b=|\Byz|\le f\), and the correct set is \(\Correct=\Nodes\setminus\Byz\). The logical communication graph is complete. Every correct receiver has an input position for every sender label in each epoch.

A Byzantine node can omit events, emit arbitrary event patterns, and choose different patterns for different receivers. It cannot forge the label of a correct node. This is immediate-receiver source attribution, not a transferable digital signature. A compatible lower layer must bind each admitted label to its physical or logical sender, for example through a trusted router port, verified address metadata, or an endpoint tag. Sybil attacks~\cite{douceur2002sybil, bakar2023review} and sender impersonation are outside the model.

Complete logical visibility does not require a physical all-to-all wire. A
routing layer can implement multicast or an overlay. If some correct sender is
not visible to some correct receiver, that missing visibility must either be
repaired below \SIF or handled by a graph-resilient agreement model. The latter
extension is outside the scope of this paper.

\subsection{Paced epochs, labels, and \FirstSpike admission}

Time is divided into logical epochs of duration \(T\). A lower layer supplies epoch identifiers and epoch boundaries whose remaining skew is included in the timestamp-error bound below. Correct senders schedule their single control event from the local state available at the epoch boundary, and receivers update after the epoch admission window closes. The agreement algorithm is event-driven within an epoch, but the model is not fully asynchronous. A single time coordinate has a common meaning only after the endpoints share a bounded-skew epoch reference.

Control events are distinguished from application spikes by a reserved event type or logical channel. An exclusive maintenance window is an equivalent implementation.

At receiver \(i\), each pair consisting of epoch \(r\) and sender label \(j\) has one admission record. The receiver accepts the earliest event carrying that pair and ignores every later event with the same pair. We call this rule \FirstSpike. If no event from label \(j\) is accepted, the receiver inserts a fixed sentinel \(\sigma\in[0,1]\). Since correct events are guaranteed to arrive in their intended epoch, silence can occur only at a Byzantine label under the stated model.

\subsection{Timing model and guards}
\label{sec:timing-model}

A correct node stores a scalar \(x_i(r)\in[0,1]\) at the start of epoch \(r\). Let \(g\) be the guard reserved at both epoch boundaries and let $L=T-2g>0$ be the usable encoding length. Correct node \(j\) emits its logical control spike at
\begin{equation}
  \tau_j(r)=rT+g+Lx_j(r).
  \label{eq:encoding-time}
\end{equation}

A correct receiver \(i\) records
\begin{equation}
  \widehat t_{ij}(r)=\tau_j(r)+d_{ij}(r)+e_{ij}(r).
  \label{eq:observation}
\end{equation}
The delivery delay satisfies \(d_{ij}(r)\in[0,D]\). This bound is required to hold despite faulty traffic. It therefore presupposes lower-layer traffic isolation, rate control, or reserved capacity. \FirstSpike limits endpoint multiplicity but does not protect the fabric before admission. The timestamp term satisfies \(|e_{ij}(r)|\le\eta\) and includes measurement jitter, timestamp quantization, and residual epoch skew. We require
\begin{equation}
  0<g<T/2,
  \qquad D+\eta<g.
  \label{eq:guard}
\end{equation}
Then every correct observation remains strictly inside its intended epoch. The guard condition concerns the full physical delay. It is distinct from the smaller residual uncertainty used for value decoding.

Receiver \(i\) knows a nominal delay \(\bar d_{ij}\) and a residual bound
\begin{equation}
  |d_{ij}(r)-\bar d_{ij}|\le\delta_d.
  \label{eq:delay-residual}
\end{equation}
We use one global \(\delta_d\) for readability. Link-specific bounds can be replaced by their maximum. Define
\begin{equation}
  \omega=\delta_d+\eta.
  \label{eq:omega}
\end{equation}
If no link calibration is available beyond \(d_{ij}(r)\in[0,D]\), the receiver chooses \(\bar d_{ij}=D/2\), which gives \(\delta_d=D/2\) and \(\omega=D/2+\eta\).

Table~\ref{tab:interface} maps each classical MSR input requirement to the corresponding spike-time problem and the \SIF mechanism that resolves it.

\begin{table}[ht]
\centering
\caption{The interface between spike-time communication and sender-indexed MSR aggregation.}
\label{tab:interface}
{\small
\begin{tabularx}{\textwidth}{>{\raggedright\arraybackslash}p{0.21\textwidth} >{\raggedright\arraybackslash}X >{\raggedright\arraybackslash}X}
\toprule
MSR input requirement & Spike-time failure & \SIF mechanism \\
\midrule
One value from each sender & A faulty label can emit many events & \FirstSpike admits at most one event per label, receiver, and epoch \\
\addlinespace
Sender attribution & An unlabeled event can be duplicated or impersonated & The lower layer supplies a nonforgeable immediate-receiver label \\
\addlinespace
Comparable epochs & Local event times need not share an origin & A paced epoch service supplies bounded-skew boundaries and epoch identifiers \\
\addlinespace
Bounded correct values & Delay and timestamp error perturb a time code & Guarded encoding, nominal-delay subtraction, and minimax decoding give error at most \(\rho\) \\
\addlinespace
A fixed input size & A faulty sender can remain silent & A missing label is represented by a local sentinel in \([0,1]\) \\
\bottomrule
\end{tabularx}
}
\end{table}

\subsection{Normal starts and agreement-layer recovery}

A normal execution satisfies all preceding \SIF assumptions from epoch \(0\). It begins with a correct input \(u_i\in[0,1]\) at every correct node and sets \(x_i(0)=u_i\). For each epoch define
\begin{equation}
  m(r)=\min_{i\in\Correct}x_i(r),
  \qquad
  M(r)=\max_{i\in\Correct}x_i(r),
  \qquad
  R(r)=M(r)-m(r).
  \label{eq:range}
\end{equation}

\begin{definition}[Robust spike-time approximate agreement]
\label{def:agreement}
A protocol achieves robust approximate agreement with decoding radius \(\rho\) and asymptotic range bound \(\Phi_\rho\) if the following properties hold.

\begin{enumerate}
\item In a normal execution, every correct node starts from its input.

\item Every correct update in a clean epoch satisfies
\begin{equation}
  x_i(r+1)\in [m(r)-\rho,M(r)+\rho]\cap[0,1].
  \label{eq:robust-validity-def}
\end{equation}

\item The correct range satisfies
\begin{equation}
  \limsup_{r\to\infty}R(r)\le\Phi_\rho,
  \qquad \Phi_0=0.
  \label{eq:agreement-bound-def}
\end{equation}
\end{enumerate}
\end{definition}

When \(\rho=0\), Equation~\eqref{eq:robust-validity-def} is the usual validity condition. It keeps every correct state in the convex hull of the original correct inputs for all time. For positive \(\rho\), the condition is one-step input-relatedness. It is strictly weaker than global validity with respect to the original correct-input hull because a new bounded error may enter in every epoch. When \(\rho<1/2\), the condition prevents an immediate jump from an all-zero or all-one correct state to \(1/2\), but it does not prevent slow common-mode drift. The sharper finite-horizon bound for \SpikeTrim and its limitation are stated after Lemma~\ref{lem:validity}.

We also allow transient corruption of the agreement layer. Before some time \(t_0\), correct scalar states, local epoch counters, and admission buffers may be arbitrary. We call an epoch \(r\) that begins at or after \(t_0\) \emph{clean} if, during that epoch, the lower-layer timing, sender-identity, and epoch services satisfy the preceding \SIF assumptions, every per-sender admission record for \(r\) has been reset before admission begins, and no event created before \(t_0\) can be admitted. Let \(r_0\) be the first clean epoch, and assume that every epoch \(r\ge r_0\) is clean. Epoch tags can enforce the last condition. Without tags, it is sufficient to wait until all pre-\(t_0\) events have drained. The delivery-delay bound \(D\) makes this waiting time finite. At the start of \(r_0\), every state is clipped to \([0,1]\).

The recovery guarantee is conditional and layer-specific. It starts from the state vector \(x(r_0)\) at the first clean epoch. It cannot reconstruct pre-corruption inputs that an arbitrary state fault may have erased. The protocol does not self-stabilize the clock service, sender labels, or physical network.

The classical resilience threshold remains necessary even with perfect timing.

\begin{theorem}[Necessary fault threshold]
\label{thm:necessity}
Let \(f\ge1\), \(2\le n\le3f\), and \(0<\eps<1\). No deterministic protocol in the \SIF communication model can simultaneously satisfy ordinary validity with \(\rho=0\) and eventually achieve \(R(r)\le\eps\) against at most \(f\) Byzantine nodes.
\end{theorem}

\begin{proof}
The proof is the standard three-set indistinguishability construction adapted to sender-labeled spike observations.
Because \(n\le3f\) and \(n\ge2\), partition \(\Nodes\) into disjoint sets \(A\), \(B\), and \(C\) such that \(A\ne\varnothing\), \(C\ne\varnothing\), and each set has size at most \(f\). The middle set can be empty. Assume for contradiction that a deterministic protocol satisfies ordinary validity and eventually reaches range at most \(\eps<1\).

Consider three executions with perfect timing and no transient corruption.

In execution \(E^0\), nodes in \(A\cup B\) are correct and start at zero. Nodes in \(C\) are Byzantine.

In execution \(E^1\), nodes in \(B\cup C\) are correct and start at one. Nodes in \(A\) are Byzantine.

In execution \(E^*\), nodes in \(A\cup C\) are correct, nodes in \(B\) are Byzantine, the nodes in \(A\) start at zero, and the nodes in \(C\) start at one.

We construct the Byzantine behavior inductively over epochs. Assume all local histories through the start of epoch \(r\) have been fixed. By the epoch-based \SIF{} semantics, determinism fixes the event scheduled by every correct node in that epoch.

In \(E^*\), every Byzantine node in \(B\) sends to a receiver in \(A\) exactly the sender-labeled observations that the corresponding correct node in \(B\) sends in \(E^0\). It sends to a receiver in \(C\) the observations that the corresponding correct node in \(B\) sends in \(E^1\).

In \(E^0\), every Byzantine node in \(C\) sends to receivers in \(A\cup B\) the observations that the corresponding correct node in \(C\) sends in \(E^*\).

In \(E^1\), every Byzantine node in \(A\) sends to receivers in \(B\cup C\) the observations that the corresponding correct node in \(A\) sends in \(E^*\).

This behavior is legal. Byzantine nodes use only their own labels, but they may choose arbitrary receiver-dependent event times or silence. By induction, every node in \(A\) has the same local history in \(E^0\) and \(E^*\). Its initial state is zero in both executions. Its observations from \(B\) and \(C\) are identical by construction. The symmetric argument shows that every node in \(C\) has the same local history in \(E^1\) and \(E^*\).

In \(E^0\), all correct initial values are zero. Ordinary validity forces every correct state to remain zero. In \(E^1\), all correct initial values are one, so every correct state remains one. Indistinguishability and determinism therefore force every correct node in \(A\) to remain zero in \(E^*\), and every correct node in \(C\) to remain one.

Both \(A\) and \(C\) are nonempty and correct in \(E^*\). Hence \(R(r)=1\) in every epoch, contradicting eventual \(\eps\)-agreement for \(\eps<1\).
\end{proof}

Theorem~\ref{thm:necessity} shows that pacing and sender attribution do not remove the \(3f+1\) threshold because Byzantine senders can still equivocate across receivers.

\section{One-spike feasible-set decoding}
\label{sec:decoder}
The timing model of Section~\ref{sec:model} reduces each accepted correct-origin timestamp, after removing the epoch origin, the guard offset, and the nominal link delay, to one noisy observation of a scalar in \([0,1]\). Before this observation can be supplied to an MSR aggregator, two questions must be settled. The first is how it should be decoded. The second is what worst-case precision is fundamentally achievable from one spike in one epoch. We therefore treat decoding as a standalone minimax estimation problem under the residual uncertainty bound \(|z|\le\omega\).

We characterize the complete interval of source values consistent with an observation and use its midpoint as the estimate. We compute its exact deterministic minimax radius \(\rho=\min\{1/2,\omega/L\}\), give a matching channel-specific lower bound, and then specialize the result to the uncalibrated one-sided delay model \(d\in[0,D]\). The resulting radius \(\rho\) is the value-domain timing guarantee passed to the interface transfer theorem and the subsequent validity, convergence, and recovery analysis.

The receiver first converts an accepted timestamp into a normalized observation
\begin{equation}
  s_{ij}(r)=\widehat t_{ij}(r)-rT-g-\bar d_{ij}.
  \label{eq:normalized-observation}
\end{equation}
For a correct sender, Equations~\eqref{eq:encoding-time}--\eqref{eq:omega} give
\begin{equation}
  s_{ij}(r)=Lx_j(r)+z_{ij}(r),
  \qquad |z_{ij}(r)|\le\omega.
  \label{eq:centered-channel}
\end{equation}

For a real number \(v\), define $\clip(v,0,1)=\min\{1,\max\{0,v\}\}$. For every feasible observation \(s\in[-\omega,L+\omega]\), the set of values consistent with Equation~\eqref{eq:centered-channel} is
\begin{equation}
  I(s)=\left[
  \clip\!\left(\frac{s-\omega}{L},0,1\right),
  \clip\!\left(\frac{s+\omega}{L},0,1\right)
  \right].
  \label{eq:feasible-interval}
\end{equation}
The decoder returns the midpoint of this interval. We use the same clipped formula for arbitrary accepted timestamps.
\begin{equation}
  \Dec(s)=\frac12\left[
  \clip\!\left(\frac{s-\omega}{L},0,1\right)
  +
  \clip\!\left(\frac{s+\omega}{L},0,1\right)
  \right].
  \label{eq:decoder}
\end{equation}
Thus a Byzantine timestamp is always mapped to a value in \([0,1]\).

Under bounded set-membership uncertainty, a Chebyshev center minimizes worst-case error and the optimum equals the Chebyshev radius~\cite{Traub1988IBC,MilaneseVicino1991}. Since \(I(s)\) is an interval, its Chebyshev center is its midpoint. The following theorem computes the resulting radius for the affine spike-time channel and gives a matching channel-specific indistinguishability construction.

\begin{theorem}[Exact minimax decoding radius]
\label{thm:minimax}
For the one-spike channel \(s=Lx+z\), where \(x\in[0,1]\) and \(|z|\le\omega\), the decoder in Equation~\eqref{eq:decoder} satisfies
\begin{equation}
  |\Dec(s)-x|\le\rho,
  \qquad
  \rho=\min\!\left\{\frac12,\frac{\omega}{L}\right\}.
  \label{eq:rho}
\end{equation}
No deterministic decoder from one observation \(s\) to an estimate in \([0,1]\) has a smaller worst-case absolute error over the same uncertainty class.
\end{theorem}

\begin{proof}
The true value belongs to \(I(s)\). The interval has width at most
\[
  \min\!\left\{1,\frac{2\omega}{L}\right\}.
\]
Its midpoint is therefore at distance at most half this width from every feasible value. This proves the upper bound.

The lower bound is immediate when \(\omega=0\). Now suppose \(\omega>0\) and \(\omega/L\le1/2\). The two distinct values \(x_0=0\) and \(x_1=2\omega/L\) produce the same observation \(s=\omega\) under errors \(z_0=\omega\) and \(z_1=-\omega\). A deterministic decoder returns one estimate for both executions, so its error is at least \(\omega/L\) in one of them.

Now suppose \(\omega/L\ge1/2\). Values zero and one produce the same observation \(s=L/2\) under errors \(L/2\) and \(-L/2\). The decoder error is at least \(1/2\) for one of these values. Both cases match Equation~\eqref{eq:rho}.
\end{proof}

Figure~\ref{fig:decoder} illustrates the feasible interval and the midpoint rule.

\begin{figure}[t]
  \centering
  \includegraphics[width=\linewidth]{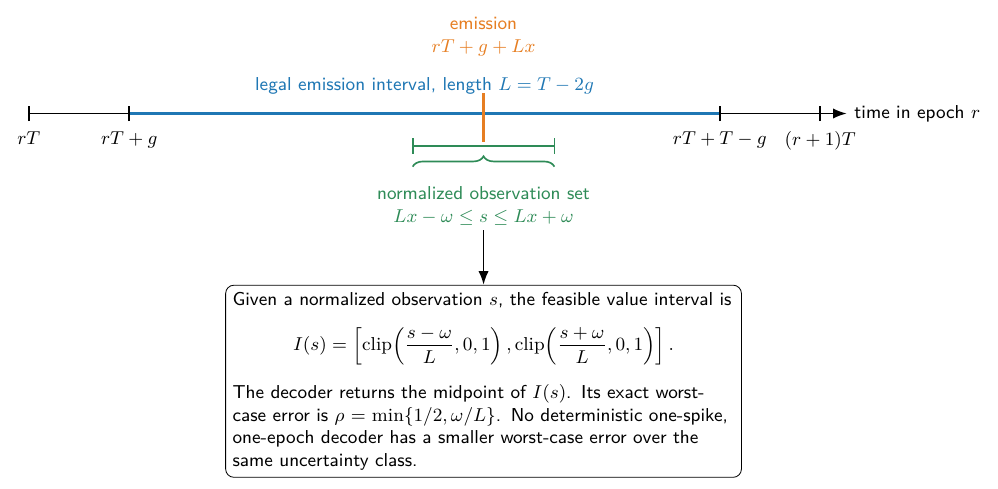}
  \caption{Guarded time encoding and interval decoding. After nominal-delay subtraction, a correct observation lies within \(\omega\) of \(Lx\). The midpoint of the feasible value interval attains the exact minimax radius \(\rho\).}
  \label{fig:decoder}
\end{figure}

\begin{corollary}[Uncalibrated one-sided delay]
\label{cor:uncalibrated}
If the receiver knows only \(d\in[0,D]\), then choosing \(\bar d=D/2\) gives
\begin{equation}
  \rho=\min\!\left\{\frac12,\frac{D/2+\eta}{T-2g}\right\}.
  \label{eq:uncalibrated-rho}
\end{equation}
The corresponding indistinguishability diameter in emission time is \(\min\{T-2g,D+2\eta\}\). The value-domain diameter is \(2\rho\).
\end{corollary}

When \(\omega\ge L/2\), the one-spike observation cannot distinguish the two endpoint values in the worst case and \(\rho=1/2\). A useful design should therefore operate in the nonsaturated regime \(\omega<L/2\), and its full asymptotic range bound should be checked with Corollary~\ref{cor:design-budget}.

The theorem isolates the exact scope of the timing lower bound. Multiple spikes can spend a larger communication budget. Multiple epochs can code information jointly. Stochastic assumptions can also improve average error. None of these possibilities contradicts the one-spike, one-epoch minimax result.

\section{The \SpikeTrim protocol}
\label{sec:protocol}
At every clean epoch, each correct node emits one encoded spike and constructs one local value for each sender label. It sets its own entry directly to its local state rather than transmitting that value through the physical channel. For any missing remote label, it inserts the sentinel \(\sigma\). The node then applies MSR and uses a relaxation parameter \(\beta\in(0,1]\) in its state update.

\begin{algorithm}[h]
\caption{\SpikeTrim at correct node \(i\) in epoch \(r\)}
\label{alg:spiketrim}
\begin{algorithmic}[1]
\State \(x_i(r)\gets\clip(x_i(r),0,1)\)
\State Reset the per-sender admission records for epoch \(r\)
\State Emit one event at time \(rT+g+Lx_i(r)\)
\State Set \(y_{ii}(r)\gets x_i(r)\)
\State Wait until the epoch admission window closes
\For{each sender label \(j\ne i\)}
  \If{an event with label \(j\) is accepted by \FirstSpike}
    \State Compute \(s_{ij}(r)\) by Equation~\eqref{eq:normalized-observation}
    \State Set \(y_{ij}(r)\gets\Dec(s_{ij}(r))\)
  \Else
    \State Set \(y_{ij}(r)\gets\sigma\)
  \EndIf
\EndFor
\State Sort \(Y_i(r)=\{y_{i1}(r),\ldots,y_{in}(r)\}\)
\State Remove its \(f\) smallest and \(f\) largest entries
\State Let \(\mu_i(r)\) be the mean of the remaining \(n-2f\) entries
\State \(x_i(r+1)\gets\clip((1-\beta)x_i(r)+\beta\mu_i(r),0,1)\)
\end{algorithmic}
\end{algorithm}

The sentinel can be any fixed value in \([0,1]\). Correctness does not depend on choosing a particular endpoint because every missing label is already counted among the at most \(f\) arbitrary sender-label entries. The use of a local self value is the standard send-to-self convention and removes unnecessary loopback timing error.

The following theorem is the interface result. It is independent of the later contraction proof.

\begin{theorem}[\SIF to MSR transfer]
\label{thm:interface}
In every clean epoch and at every correct receiver \(i\), the multiset \(Y_i(r)\) contains exactly one value for each sender label. For every correct label \(j\),
\begin{equation}
  |y_{ij}(r)-x_j(r)|\le\rho.
  \label{eq:correct-origin-bound}
\end{equation}
The remaining \(b\le f\) entries correspond to Byzantine labels. They are arbitrary values in \([0,1]\) and may differ across correct receivers.
\end{theorem}

\begin{proof}
The local entry is exact. For every other correct label, Equation~\eqref{eq:guard} keeps its event in the intended epoch. Identity preservation prevents a Byzantine node from replacing that entry. Theorem~\ref{thm:minimax} gives Equation~\eqref{eq:correct-origin-bound}.

For a Byzantine label, \FirstSpike admits at most one event. If no event is admitted, the sentinel supplies one entry. The clipped decoder and the sentinel both lie in \([0,1]\). There are exactly \(b\) Byzantine labels.
\end{proof}

The theorem exposes the modularity of the design. Above \SIF, the receiver sees
the usual sender-indexed Byzantine fault model, augmented by a bounded error on
correct-origin values. Other robust aggregators that accept one bounded value
per sender under this same fault-and-error model can likewise use \SIF,
although this paper analyzes only the classical trimmed mean.

\begin{lemma}[Robust validity]
\label{lem:validity}
Assume \(n\ge3f+1\). For every clean epoch and every correct node \(i\), \SpikeTrim satisfies the sharper bound
\begin{equation}
  x_i(r+1)\in[m(r)-\beta\rho,M(r)+\beta\rho]\cap[0,1].
  \label{eq:beta-validity}
\end{equation}
Hence it satisfies the robust-validity condition in Definition~\ref{def:agreement}.
\end{lemma}

\begin{proof}
Every correct-origin decoded value belongs to \([m(r)-\rho,M(r)+\rho]\). A value below this interval must come from a Byzantine label. There are at most \(f\) such labels, so all values below the interval are removed by the low trim. The symmetric argument applies above the interval. Therefore
\[
  \mu_i(r)\in[m(r)-\rho,M(r)+\rho].
\]
The unclipped update is a convex combination of \(x_i(r)\in[m(r),M(r)]\) and \(\mu_i(r)\). It lies in the interval in Equation~\eqref{eq:beta-validity}. Clipping to \([0,1]\) cannot move it outside the intersection.
\end{proof}

For a normal start, Lemma~\ref{lem:validity} gives
\begin{equation}
  x_i(r)\in[m(0)-r\beta\rho,M(0)+r\beta\rho]\cap[0,1].
  \label{eq:beta-input-envelope}
\end{equation}
When \(\rho=0\), every correct state remains in the original correct input hull. This excludes rules that ignore the inputs and always return \(1/2\).

\begin{remark}[Cumulative location drift]
\label{rem:cumulative-drift}
For \(\rho>0\), robust validity is only a one-step input-related property. Equation~\eqref{eq:beta-input-envelope} gives the finite-horizon bound
\begin{equation}
  \max_{i\in\Correct}
  \operatorname{dist}\!\left(x_i(r),[m(0),M(0)]\right)
  \le \min\{1,r\beta\rho\}.
  \label{eq:cumulative-location-drift}
\end{equation}
The dependence on \(r\) is not merely a proof artifact. Let \(f\ge1\), \(b=f\), and \(0<\rho<1/2\), and suppose that all correct states equal an interior value \(c_r\). At each receiver, choose residual timing error \(+\omega\) for every remote correct-origin observation. Whenever \(c_r+2\rho\le1\), every such observation decodes to \(c_r+\rho\). Let the local self value remain exact and let every Byzantine entry equal one. The low trim removes the exact self value and \(f-1\) remote values. The high trim removes the \(f\) Byzantine values. All \(n-2f\) retained entries then equal \(c_r+\rho\), so every correct state moves to \(c_{r+1}=c_r+\beta\rho\). Thus admissible same-direction timing errors can translate an already narrow correct-state interval while its range remains small.

The asymptotic range bound controls separation among correct nodes. It does not
control their common bias relative to the initial correct-input hull. For
applications that periodically obtain fresh local reference estimates,
reseeding the tile states before a finite maintenance phase provides one way to
limit common-mode drift. If such a phase uses \(K\) updates, its additional
outward displacement is bounded by \(\min\{1,K\beta\rho\}\). Periodic
re-estimation is an application-level mitigation, not a global-validity
property of \SpikeTrim.
\end{remark}

\section{Convergence and agreement-layer self-stabilization}
\label{sec:analysis}
Section~\ref{sec:protocol} establishes the one-step safety side of \SpikeTrim. After the \SIF adapter, every correct receiver has one bounded value per sender label, and the MSR update satisfies robust validity. The remaining issue is dynamic. Receiver-dependent Byzantine values can cause different correct receivers to retain different subsets of the sender-indexed values. A new decoding error may also enter in every epoch. Robust validity alone therefore guarantees neither contraction of the correct-state range nor recovery after transient corruption.

This section closes that gap by deriving a common interval that contains every correct receiver's trimmed mean. Its width is controlled by the current correct-state range. This interval recovers the classical tight noiseless contraction factor \(f/(n-2f)\), yields the noisy range recurrence for every \(\beta\in(0,1]\), the explicit asymptotic range bound \(\Phi_\rho\), and a finite recovery bound. We then rearrange the bound to obtain a timing-budget condition for a prescribed disagreement target.

The self-stabilization guarantee is intentionally agreement-layer specific. It starts at the first clean post-corruption epoch and assumes that the lower-layer \SIF{} timing, identity, epoch, and admission mechanisms satisfy their stated assumptions. It does not reconstruct inputs erased by the transient fault.

We retain the global epoch index. Let \(r_0\) be the first clean post-corruption epoch. For recovery statements, \(k\ge0\) counts updates in the clean suffix, so the corresponding global epoch is \(r_0+k\). A normal execution has \(r_0=0\).

Let \(p=n-2f\) be the number of retained entries. Sort the true correct values as
\[
  v_1(r)\le v_2(r)\le\cdots\le v_{n-b}(r).
\]
Define
\begin{align}
  \ell(r)
  &=\frac1p\sum_{k=f-b+1}^{f-b+p}v_k(r),
  \label{eq:ell}\\
  u(r)
  &=\frac1p\sum_{k=f+1}^{f+p}v_k(r).
  \label{eq:u}
\end{align}
All indices are valid because \(b\le f\) and \(n\ge3f+1\).

\begin{lemma}[Common interval for all trimmed means]
\label{lem:trim-interval}
For every correct receiver \(i\),
\begin{equation}
  \mu_i(r)\in[\ell(r)-\rho,u(r)+\rho].
  \label{eq:trim-interval}
\end{equation}
Moreover,
\begin{equation}
  u(r)-\ell(r)
  \le \frac{b}{n-2f}R(r)
  \le \frac{f}{n-2f}R(r).
  \label{eq:interval-width}
\end{equation}
\end{lemma}

\begin{proof}
Fix a correct receiver and omit the epoch argument. Sort its correct-origin decoded values as \(w_1\le\cdots\le w_{n-b}\). Sort all \(n\) entries, including the \(b\) Byzantine-label entries, as \(a_1\le\cdots\le a_n\). The \(h\)-th retained value is \(a_{f+h}\), where \(1\le h\le p\).

Fewer than \(f+h\) entries can be strictly smaller than \(w_{f-b+h}\). At most \(f-b+h-1\) are correct-origin entries and at most \(b\) are Byzantine-label entries. Hence
\[
  a_{f+h}\ge w_{f-b+h}.
\]
At least \(f+h\) correct-origin entries are no larger than \(w_{f+h}\), so
\[
  a_{f+h}\le w_{f+h}.
\]
The componentwise decoding bound implies \(v_k-\rho\le w_k\le v_k+\rho\) for every order statistic \(k\). Averaging the preceding inequalities over \(h=1,\ldots,p\) proves Equation~\eqref{eq:trim-interval}.

To bound its width, cancel the overlapping terms in Equations~\eqref{eq:ell} and \eqref{eq:u}. The difference is
\[
  \frac1p\left(
  \sum_{k=n-f-b+1}^{n-f}v_k
  -
  \sum_{k=f-b+1}^{f}v_k
  \right).
\]
Each sum has \(b\) terms. Every term in the first sum is at most \(M(r)\), and every term in the second is at least \(m(r)\). This gives Equation~\eqref{eq:interval-width}.
\end{proof}

\subsection{Tight contraction without timing error}

The all-retained update is a member of the classical MSR family~\cite{Kieckhafer1994MixedMode}. We restate its noiseless factor in the present notation because it supplies the contraction constant for the bounded-error analysis. The construction below also verifies attainability under receiver-dependent Byzantine values.

\begin{theorem}[Tight noiseless contraction]
\label{thm:tight-contraction}
Assume \(\rho=0\), \(\beta=1\), and \(n\ge3f+1\). Then
\begin{equation}
  R(r+1)\le\frac{f}{n-2f}R(r).
  \label{eq:tight-contraction}
\end{equation}
The factor is tight for every \(f\ge1\).
\end{theorem}

\begin{proof}
Lemma~\ref{lem:trim-interval} places every correct output in one interval of width at most \(fR(r)/(n-2f)\). This proves the upper bound.

For tightness, let the actual number of Byzantine nodes be \(b=f\). Choose two values \(0<a<c<1\). Give \(f\) correct nodes value \(a\) and the remaining \(n-2f\) correct nodes value \(c\). To one correct receiver, every Byzantine sender reports zero. To another, every Byzantine sender reports one. The first receiver retains \(f\) copies of \(a\) and \(n-3f\) copies of \(c\). The second retains only copies of \(c\). Their output difference is
\[
  \frac{f}{n-2f}(c-a),
\]
which attains Equation~\eqref{eq:tight-contraction}.
\end{proof}

If \(f=0\) and \(\beta=1\), all correct receivers compute the same mean in one epoch. For \(f\ge1\), iteration gives
\begin{equation}
  R(r)\le\left(\frac{f}{n-2f}\right)^r R(0).
  \label{eq:noiseless-iteration}
\end{equation}
Thus a normal execution reaches \(R(r)\le\eps\) after at most
\begin{equation}
  \left\lceil
  \frac{\log(R(0)/\eps)}{\log((n-2f)/f)}
  \right\rceil
  \label{eq:noiseless-rounds}
\end{equation}
epochs whenever \(R(0)>\eps\).

\subsection{Bounded timing error and convex updates}

\begin{theorem}[Range recurrence]
\label{thm:range-recurrence}
For \(n\ge3f+1\), \(\beta\in(0,1]\), and every clean epoch,
\begin{equation}
  R(r+1)
  \le
  \left(1-\beta+\beta\frac{b}{n-2f}\right)R(r)
  +2\beta\rho.
  \label{eq:actual-b-recurrence}
\end{equation}
In the worst case \(b=f\), define
\begin{equation}
  q_\beta=1-\beta\frac{n-3f}{n-2f}.
  \label{eq:qbeta}
\end{equation}
Then
\begin{equation}
  R(r+1)\le q_\beta R(r)+2\beta\rho.
  \label{eq:worst-recurrence}
\end{equation}
\end{theorem}

\begin{proof}
The unclipped update at correct node \(i\) is
\[
  z_i(r+1)=(1-\beta)x_i(r)+\beta\mu_i(r).
\]
All states and aggregates lie in \([0,1]\), so clipping is inactive. For two correct nodes \(i\) and \(k\),
\[
  |z_i(r+1)-z_k(r+1)|
  \le(1-\beta)R(r)+\beta|\mu_i(r)-\mu_k(r)|.
\]
Lemma~\ref{lem:trim-interval} bounds the aggregate difference by
\[
  \frac{b}{n-2f}R(r)+2\rho.
\]
Taking the maximum over correct pairs proves Equation~\eqref{eq:actual-b-recurrence}. Substituting \(b\le f\) gives Equations~\eqref{eq:qbeta} and \eqref{eq:worst-recurrence}.
\end{proof}

For any clean base epoch \(r_{\mathrm b}\) and any integer \(k\ge0\), unrolling the recurrence yields
\begin{equation}
  R(r_{\mathrm b}+k)
  \le q_\beta^k R(r_{\mathrm b})
  +\frac{2\rho(n-2f)}{n-3f}\left(1-q_\beta^k\right).
  \label{eq:unrolled}
\end{equation}
The state range is always at most one. Define the worst-case asymptotic range bound
\begin{equation}
  \Phi_\rho
  =\min\!\left\{1,\frac{2\rho(n-2f)}{n-3f}\right\}.
  \label{eq:asymptotic-bound}
\end{equation}

\begin{theorem}[Quantitative agreement-layer self-stabilization]
\label{thm:self-stabilization}
Assume \(n\ge3f+1\). From any first clean post-corruption state vector, \SpikeTrim satisfies robust validity and
\begin{equation}
  \limsup_{k\to\infty}R(r_0+k)\le\Phi_\rho.
  \label{eq:limsup}
\end{equation}
For any \(\eps>0\), define \(k_\eps=0\) when \(R(r_0)\le\eps\). If \(R(r_0)>\eps\) and \(0<q_\beta<1\), define
\begin{equation}
  k_\eps=
  \left\lceil
  \frac{\log(R(r_0)/\eps)}{-\log q_\beta}
  \right\rceil.
  \label{eq:recovery-rounds}
\end{equation}
If \(R(r_0)>\eps\) and \(q_\beta=0\), define \(k_\eps=1\). In every case,
\begin{equation}
  R(r_0+k)\le\eps+\Phi_\rho
  \qquad\text{for all }k\ge k_\eps.
  \label{eq:finite-recovery}
\end{equation}
\end{theorem}

\begin{proof}
State clipping gives \(R(r_0)\le1\) at the first clean epoch. Lemma~\ref{lem:validity} gives robust validity independently of pre-corruption history. Put \(c=2(n-2f)/(n-3f)\). If \(c\rho\le1\), Equation~\eqref{eq:unrolled} with \(r_{\mathrm b}=r_0\) gives
\[
  R(r_0+k)\le q_\beta^kR(r_0)+\Phi_\rho.
\]
The stated definition of \(k_\eps\) makes the first term at most \(\eps\). If \(c\rho>1\), then \(\Phi_\rho=1\), and Equation~\eqref{eq:finite-recovery} follows directly from \(R(r_0+k)\le1\). These arguments also give the limsup bound. When \(q_\beta=0\), necessarily \(f=0\) and \(\beta=1\). Hence \(R(r_0+1)\le2\rho=\Phi_\rho\), including the endpoint case \(\rho=1/2\).
\end{proof}

The finite recovery bound is informative when \(\Phi_\rho<1\).

The worst-case asymptotic range bound in Equation~\eqref{eq:asymptotic-bound} does not depend on \(\beta\). A smaller \(\beta\) slows the provable transient contraction and reduces the additive error injected in each individual step by the same proportion. It can provide implementation inertia against unmodeled stochastic effects, but it does not improve the stated adversarial bound. Therefore the experiments use \(\beta=1\).

\begin{corollary}[Timing budget for a target asymptotic bound]
\label{cor:design-budget}
Assume \(n\ge3f+1\), and let \(0<\tau<1\) be a target upper bound on asymptotic correct-state range. The proved bound in Equation~\eqref{eq:asymptotic-bound} satisfies \(\Phi_\rho\le\tau\) if and only if
\begin{equation}
  \rho\le\frac{\tau(n-3f)}{2(n-2f)}.
  \label{eq:rho-budget}
\end{equation}
The threshold on the right-hand side of Equation~\eqref{eq:rho-budget} is strictly smaller than \(1/2\). Since
\(\rho=\min\{1/2,\omega/L\}\), the condition can hold only in the nonsaturated decoder regime \(\omega<L/2\), where \(\rho=\omega/L\). It is therefore equivalent to
\begin{equation}
  L=T-2g
  \ge
  \frac{2\omega(n-2f)}{\tau(n-3f)}.
  \label{eq:L-budget}
\end{equation}
\end{corollary}

Corollary~\ref{cor:design-budget} is necessary and sufficient for the derived bound \(\Phi_\rho\) to meet the target \(\tau\). If the condition fails, the present worst-case analysis cannot establish that target for \SpikeTrim. This is not an impossibility theorem for every one-spike protocol.

For fixed \(T\), \(D\), \(\eta\), and \(\omega\), suppose the guard \(g\) can be selected. Equations~\eqref{eq:guard} and \eqref{eq:L-budget} can be satisfied simultaneously if and only if
\begin{equation}
  T>
  2(D+\eta)
  +\frac{2\omega(n-2f)}{\tau(n-3f)}.
  \label{eq:joint-timing-budget}
\end{equation}
In the uncalibrated case, \(\omega=D/2+\eta\). Equation~\eqref{eq:joint-timing-budget} combines the admission guard and value-precision requirements into one physical epoch-budget test.

The denominator \(n-3f\) makes the engineering tradeoff explicit. Operating close to the resilience boundary amplifies timing uncertainty and slows convergence. Adding correct redundancy or reducing residual timing error can be more effective than changing the convex update parameter.

\section{Model-based evaluation}
\label{sec:evaluation}
Our evaluation is structured around six concrete questions that assess the behavior and practical relevance of the proposed interface.

\begin{enumerate}
\item Does \SpikeTrim contract under receiver-dependent Byzantine values while simpler aggregators fail or converge more slowly?
\item Does reduced control-state disagreement also reduce disagreement in the predictions of otherwise identical spiking classifiers?
\item Does the late-epoch range grow with the decoder error radius and stay below the worst-case theorem bound?
\item Does the deterministic construction exhibit the \(3f+1\) feasibility boundary and the predicted noiseless epoch count?
\item Does per-sender \FirstSpike admission remove the influence of burst multiplicity?
\item Does the agreement layer recover after transient state corruption?
\end{enumerate}

Our simulator\footnote{The source code is hosted in our private GitHub repository that will be publicly
available upon the acceptance of the manuscript.} follows the model exactly. A correct link independently draws a delivery delay uniformly from \([0,D]\) and a timestamp error uniformly from \([-\eta,\eta]\). The receiver subtracts \(D/2\), so \(\omega=D/2+\eta\), and applies Equation~\eqref{eq:decoder}. Byzantine labels use a split-extremes strategy. They send decoded value zero to receivers in a lower group and value one to receivers in an upper group. This receiver-dependent behavior drives disagreement and is permitted by the model.

Random experiments use 200 executions generated with NumPy Generator objects and the PCG64 bit generator. Fixed experiment-level seeds are recorded in the source code. Compared stochastic methods share the same initial conditions and channel samples. The correct initial states are independent \(\operatorname{Unif}[0,1]\) samples. Curves report medians. Shaded regions show the empirical 10th to 90th percentiles where visible.
Table~\ref{tab:parameters} gives the default parameters. The time unit is normalized. The default setting has \(L=0.8\), \(\omega=0.02\), \(\rho=0.025\), and worst-case asymptotic range bound \(\Phi_\rho=0.075\). These defaults were chosen to make both the transient and timing-limited regimes visible. A \(20\%\) Byzantine fraction remains safely inside the \(3f+1\) region and, with \(\beta=1\), gives \(q_\beta=1/3\). The values \(g=0.1\), \(D=0.02\), and \(\eta=0.01\) comfortably satisfy Equation~\eqref{eq:guard} and yield a nontrivial but nonsaturated decoder radius. The paired design reduces run-to-run variation between methods, and 200 executions support the reported median and percentile summaries.

\begin{table}[h]
\centering
\caption{Default simulation parameters. Individual sweeps vary the quantities identified in the text.}
\label{tab:parameters}
\begin{tabular}{lll}
\toprule
Quantity & Value & Meaning \\
\midrule
\(n\) & 40 & sender labels \\
\(f\) & 8 & Byzantine labels and trim parameter \\
\(T\) & 1.0 & epoch duration \\
\(g\) & 0.1 & guard at each endpoint \\
\(D\) & 0.02 & maximum correct-link delay \\
\(\eta\) & 0.01 & timestamp-error radius \\
\(\beta\) & 1.0 & direct MSR update \\
Monte Carlo executions & 200 & paired executions \\
\bottomrule
\end{tabular}
\end{table}

\paragraph{Physical interpretation of the normalized times}
The model depends on ratios of physical times and is therefore scale invariant. Assigning \(T=\SI{1}{\milli\second}\) to the default point gives \(g=\SI{100}{\micro\second}\), \(D=\SI{20}{\micro\second}\), \(\eta=\SI{10}{\micro\second}\), and \(L=\SI{800}{\micro\second}\). The guard condition then has \(\SI{70}{\micro\second}\) of slack. The dimensionless values \(\rho=0.025\) and \(\Phi_\rho=0.075\) are unchanged.

Published platform timing figures place this illustrative scaling in context. At \(\SI{0.75}{\volt}\), Loihi reports pre-silicon standard-delay-format (SDF) and SPICE values of \(\SI{2.1}{\nano\second}\) for within-tile spike latency, \(\SI{4.1}{\nano\second}\) and \(\SI{6.5}{\nano\second}\) for east to west and north to south tile hops, and \(\SI{113}{\nano\second}\) to \(\SI{465}{\nano\second}\) for mesh-wide barrier synchronization from one to 32 tiles~\cite{Davies2018Loihi}. The SpiNNaker experiment reports an uncongested intra-chip round-trip delay of \(\SI{0.825}{\micro\second}\), including API software overhead, and a maximum observed round-trip delay of \(\SI{6.5}{\micro\second}\) under the tested congestion and router settings~\cite{Lagorce2015SpiNNaker}. The DYNAPs prototype reports a measured through-chip pass-through latency of \(\SI{15.4}{\nano\second}\) and a \(\SI{27}{\nano\second}\) broadcast time set from worst-case content-addressable-memory (CAM) timing assumptions~\cite{Moradi2018DYNAPs}. None of these figures is a deterministic end-to-end bound for our model, and none measures \(\eta\). A deployment must bound the full correct-route delay \(D\) at the intended load, include nominal-delay calibration residual in \(\delta_d\), and include timestamp quantization and residual epoch skew in \(\eta\).

As a conditional scale check, suppose a deployment with \(T=\SI{1}{\milli\second}\) and \(g=\SI{100}{\micro\second}\) validates \(D=\SI{6.5}{\micro\second}\) and \(\eta=\SI{1}{\micro\second}\). For \(n=40\) and \(f=8\), the uncalibrated model gives \(\rho=0.0053125\) and \(\Phi_\rho=0.0159375\). The proved bound would then meet a target \(\tau=0.02\). These assumed bounds are not inferred from the round-trip figures above. This substitution is illustrative rather than a claim about SpiNNaker under unmeasured routes or loads.

We compare four sender-indexed aggregators. \SpikeTrim uses the noisy timing channel and MSR. \emph{Ideal trimmed mean} applies the same MSR rule to exact scalar messages and isolates the timing cost. \emph{Naive mean} averages all decoded sender values without trimming. \emph{Median} takes the sender-indexed median. For even \(n\), the median is the arithmetic mean of the two central order statistics, matching the NumPy implementation used in the experiments. Trimmed means and coordinate-wise medians are standard Byzantine-robust aggregation baselines~\cite{Dolev1986ApproxAgree,yin2018byzantine}. We make no general optimality claim about median dynamics.

The baseline set answers controlled mechanism questions rather than ranking all robust aggregators. The ideal trimmed mean is the direct oracle comparison because it differs from \SpikeTrim only by replacing the timing channel with exact scalar messages. Another complete-graph MSR baseline would duplicate that rule. Sparse-graph W-MSR variants become distinct only after changing the communication topology. Krum was designed for stochastic-gradient aggregation and a learning-convergence objective rather than iterative scalar validity and range contraction~\cite{Blanchard2017Krum}. The pulse-coupled methods cited above target phase and frequency synchronization~\cite{iori2024resilient,YanIshii2025SecurePCO}. We therefore treat these methods as adjacent work rather than compare numerical results for non-equivalent tasks.

A separate ablation uses \emph{raw-event trimming}. It treats every observed event as a distinct sample while still removing only \(f\) low and \(f\) high samples. This deliberately violates the one-value-per-sender condition.

\subsection{Convergence and baseline behavior}

Figure~\ref{fig:convergence} shows the correct-node range for the default setting. All methods start from the same median range \(0.952\). \SpikeTrim reduces it to \(0.248\) after one epoch, \(0.018\) after five epochs, and approximately \(0.016\) after 50 epochs. The ideal trimmed mean reaches \(0.003\) after five epochs and numerical zero later. The difference between the two curves isolates the timing channel.

The naive mean settles near \(0.205\). The split-extremes adversary creates a persistent receiver-dependent bias because no values are removed. Under this experiment, the median baseline eventually reaches a range close to \SpikeTrim, but its transient is slower. Its median range across runs remains \(0.123\) at epoch 10 and reaches \(0.016\) near epoch 25.

\begin{figure}[t]
  \centering
  \includegraphics[width=0.75\linewidth]{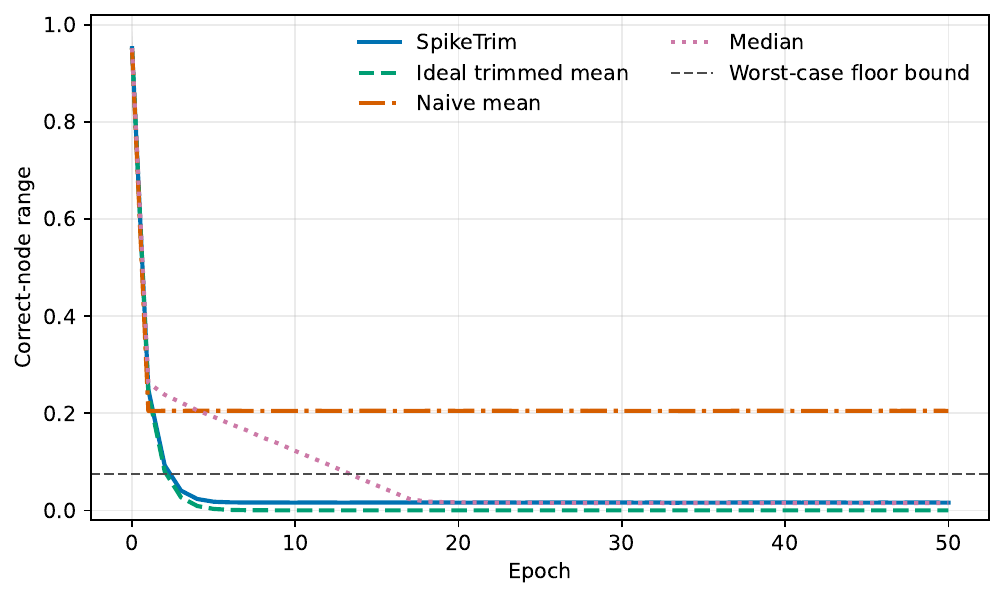}
  \caption{Paired convergence experiment with \(n=40\), \(f=8\), \(D=0.02\), and \(\eta=0.01\). The dashed horizontal line is the worst-case asymptotic range bound \(0.075\). The median across stochastic runs is lower because the simulation does not choose worst-case timing errors in every epoch.}
  \label{fig:convergence}
\end{figure}

The theorem line is an adversarial upper bound, not a prediction of the typical late-epoch range under independent noise. The observed range near \(0.016\) is therefore consistent with the bound but does not establish its tightness.

\subsection{Task-level prediction consistency}
\label{sec:task-level}

We next tested whether reducing disagreement among replicated threshold-reference
states is associated with more consistent predictions across otherwise
identical spiking classifiers. This is a controlled mechanism stress test rather than a competitive neural benchmark. The task used the 1797-sample handwritten-digits data set distributed with scikit-learn~\cite{pedregosa2011scikit}. Pixel intensities were divided by 16. We trained five current-driven spiking classifiers on stratified splits generated with seeds 1201 to 1205. Each split used 75\% training data and 25\% test data. Each test split contained 450 samples. Each network had 64 inputs, 128 hidden leaky integrate-and-fire (LIF) neurons,
and 10 output neurons. The simulation used 20 steps, membrane decay \(0.9\), and subtractive reset.
Here, learning denotes only the offline optimization of the classifier parameters, for which a sigmoid surrogate of slope 10 supplies a smooth proxy derivative at the spike discontinuity~\cite{neftci2019surrogate,ferdowsi2025silicon}. Hard threshold crossings are retained in the forward LIF simulation. Both the \SpikeTrim state evolution and the reported test-time predictions are computed directly under their stated rules.
Adam training used 35 epochs, learning rate \(10^{-3}\), weight decay \(10^{-4}\), and batch size 128. Training used the nominal threshold multiplier one. The predicted class maximized the output spike count plus \(0.05\) times the final output membrane potential.

Within each evaluation, all correct tiles used the same trained weights and test examples and differed only in their local threshold state. The unquantized nominal threshold corresponded to \(x=0.5\). For this stress test, a tile state \(x\) set the common hidden and output threshold multiplier to \(0.25+1.5x\). The task-level threshold register used 8-bit uniform quantization. This range was chosen to expose the functional effect of control-state disagreement and is not presented as a measured hardware mismatch distribution.

The no-coordination baseline left every initial tile state unchanged. Each trained network was paired with the same 200 communication executions. Mean accuracy is averaged over correct tiles. Worst accuracy is the minimum over correct tiles. Pairwise prediction disagreement is the mean fraction of test samples assigned different classes by an unordered pair of correct tiles.

The accuracy summaries are reported at epoch 10. Pairwise disagreement is reported at epochs 10 and 50. The 1000 evaluations arise from five trained networks crossed with 200 shared communication executions. They do not represent 1000 independently trained models.

Table~\ref{tab:task-level} reports pooled descriptive summaries from the paired evaluations.

\begin{table}[h]
\centering
\caption{Task-level results under the default Byzantine setting. Timing-channel methods use the default timing parameters. Entries are pooled empirical medians followed by 10th to 90th percentile ranges, all in percent. The ranges are descriptive and are not confidence intervals.}
\label{tab:task-level}
{\footnotesize
\setlength{\tabcolsep}{4pt}
\begin{tabular}{lcc}
\toprule
Method & \shortstack{Mean accuracy\\at \(r=10\)} & \shortstack{Worst-tile accuracy\\at \(r=10\)} \\
\midrule
Ideal trimmed mean & \(97.778\;(97.333\text{ to }98.889)\) & \(97.778\;(97.333\text{ to }98.889)\) \\
\SpikeTrim & \(97.694\;(97.271\text{ to }98.861)\) & \(97.556\;(97.111\text{ to }98.667)\) \\
Median & \(97.660\;(97.277\text{ to }98.785)\) & \(97.333\;(96.889\text{ to }98.667)\) \\
Naive mean & \(97.590\;(97.333\text{ to }98.785)\) & \(97.556\;(96.889\text{ to }98.667)\) \\
No coordination & \(96.750\;(95.777\text{ to }97.440)\) & \(86.222\;(79.556\text{ to }93.778)\) \\
\bottomrule
\end{tabular}

\medskip

\begin{tabular}{lcc}
\toprule
Method & \shortstack{Pairwise disagreement\\at \(r=10\)} & \shortstack{Pairwise disagreement\\at \(r=50\)} \\
\midrule
Ideal trimmed mean & \(0.000\;(0.000\text{ to }0.000)\) & \(0.000\;(0.000\text{ to }0.000)\) \\
\SpikeTrim & \(0.194\;(0.027\text{ to }0.355)\) & \(0.193\;(0.039\text{ to }0.347)\) \\
Median & \(0.334\;(0.086\text{ to }0.620)\) & \(0.185\;(0.014\text{ to }0.345)\) \\
Naive mean & \(0.422\;(0.114\text{ to }0.681)\) & \(0.396\;(0.114\text{ to }0.688)\) \\
No coordination & \(2.838\;(1.471\text{ to }4.573)\) & \(2.838\;(1.471\text{ to }4.573)\) \\
\bottomrule
\end{tabular}
}
\end{table}

At the ten-epoch budget, the pooled median pairwise disagreement is \(0.194\%\) for \SpikeTrim, \(0.334\%\) for the median rule, and \(0.422\%\) for the naive mean. The pooled median worst-tile accuracy is \(97.556\%\) for \SpikeTrim and \(97.333\%\) for the median rule. Without coordination, it is \(86.222\%\). At epoch 50, the pooled median disagreement is \(0.193\%\) for \SpikeTrim and \(0.185\%\) for the median rule. Ideal trimmed mean remains the oracle comparison because it receives exact scalar messages and incurs no timing error. These pooled summaries do not support an inferential or general superiority claim. In this stress test, faster scalar-state contraction was associated with lower pairwise prediction disagreement at the ten-epoch budget.

\subsection{Timing uncertainty}

Figure~\ref{fig:noise} varies \(\eta\) from zero to \(0.04\), which changes \(\rho\) from \(0.0125\) to \(0.0625\). For each run, the late-epoch statistic is the average range over the final ten epochs of an 80-epoch simulation. The median grows from \(0.0115\) to \(0.0475\). Every median point remains below the corresponding worst-case bound, which grows linearly from \(0.0375\) to \(0.1875\).

\begin{figure}[t]
  \centering
  \includegraphics[width=0.7\linewidth]{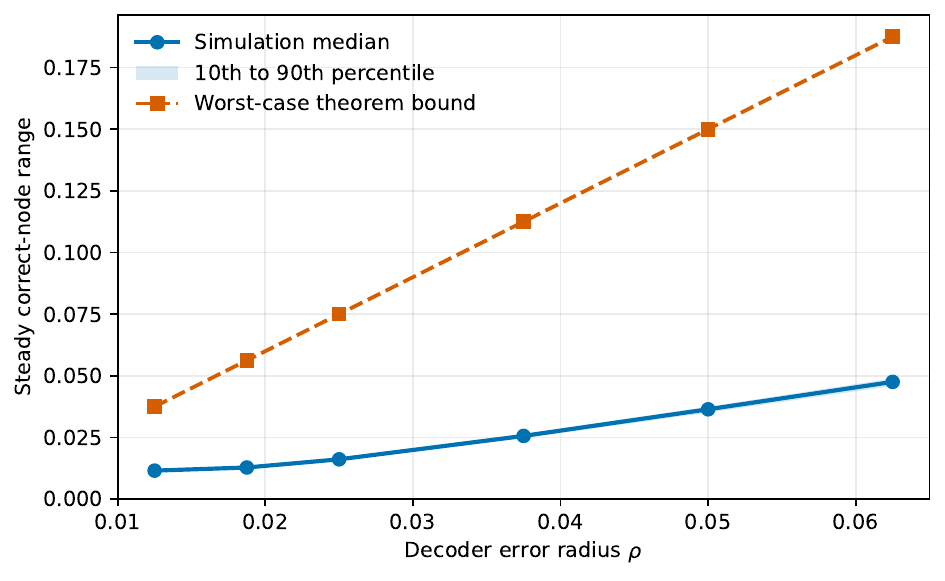}
  \caption{Late-epoch correct-node range against the exact decoder radius \(\rho\). The square-marked dashed curve is the theorem bound in Equation~\eqref{eq:asymptotic-bound}. The experiment uses independent bounded timing samples, while the theorem permits adversarial errors in every epoch.}
  \label{fig:noise}
\end{figure}

The experiment supports the predicted dependence on \(\rho\). It also shows the expected gap between a deterministic worst-case guarantee and typical independent noise.

\subsection{Fault threshold and contraction count}

To isolate the fault threshold, we set \(n=61\), remove timing noise, and vary \(f\) from zero to 25. Correct nodes start in the tight two-level construction used in Theorem~\ref{thm:tight-contraction}. The Byzantine labels send low values to the lower receiver group and high values to the upper group.

The left panel of Figure~\ref{fig:threshold} shows the range after 140 epochs. It converges throughout the feasible region \(f\le20\), where \(61\ge3f+1\), and remains one for \(f\ge21\). At \(f=20\), the contraction factor is \(20/21\), so convergence is deliberately slow. The range reaches \(10^{-2}\) after 95 epochs. The right panel compares the observed first hitting time with Equation~\eqref{eq:noiseless-rounds}. The two are identical for this tight deterministic construction.

\begin{figure}[t]
  \centering
  \includegraphics[width=0.8\linewidth]{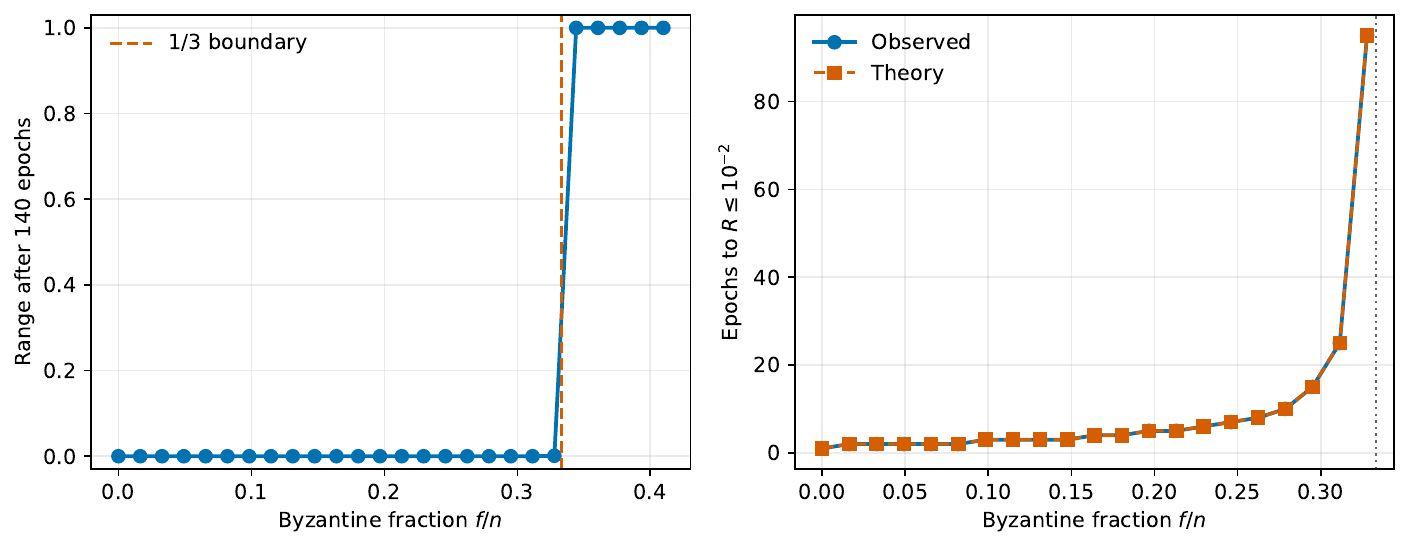}
  \caption{Deterministic threshold experiment with \(n=61\) and no timing error. Left, the correct range after 140 epochs. Right, observed epochs to \(R\le10^{-2}\) and the exact count from the tight contraction recurrence on the feasible side.}
  \label{fig:threshold}
\end{figure}

This experiment illustrates both parts of the theory. The resilience boundary is sharp, and performance deteriorates as \(n-3f\) approaches one.

\subsection{Sender flooding and \FirstSpike}

The burst ablation uses \(n=31\), \(f=8\), and noiseless timestamps. The horizontal axis counts additional events emitted by each faulty label. Thus, the first nonzero point gives two total faulty events per label and per epoch. \SpikeTrim admits only the first event. Raw-event trimming admits all of them but still trims only eight values at each end.

Figure~\ref{fig:burst} shows that \SpikeTrim reaches numerical agreement for every tested multiplicity. Raw-event trimming has range \(0.533\) with one additional event and grows to \(0.932\) with twelve additional events. The first failure point matches Observation~\ref{obs:raw-events}. The ablation isolates the role of sender-indexed admission rather than the role of trimming itself.

\begin{figure}[t]
  \centering
  \includegraphics[width=0.76\linewidth]{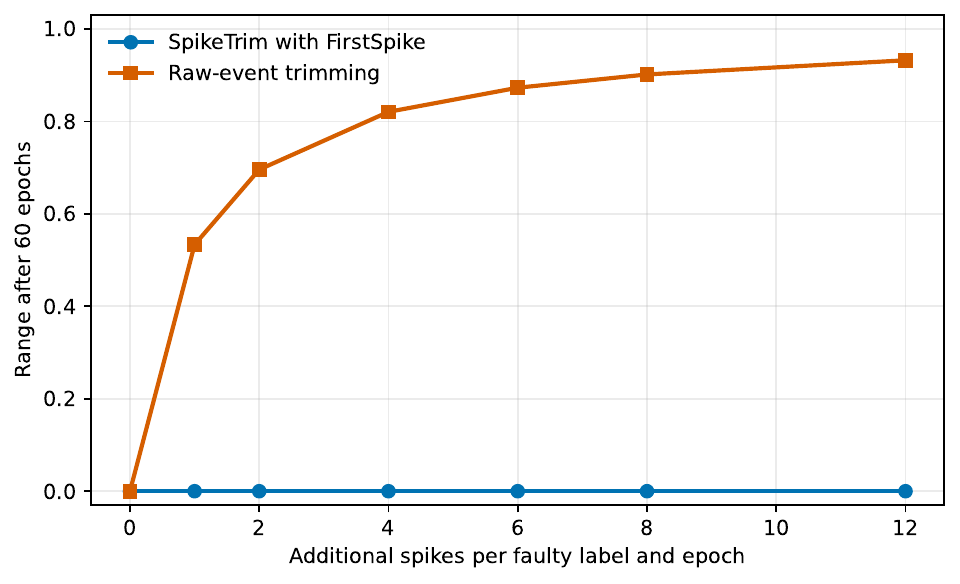}
  \caption{Burst-multiplicity ablation. \FirstSpike keeps one effective value per sender label. Raw-event trimming allows one faulty label to consume multiple positions in the aggregate and loses agreement as soon as each faulty label contributes two total events.}
  \label{fig:burst}
\end{figure}

\subsection{Recovery after transient corruption}

Finally, we corrupt 13 of the 32 correct states, approximately 40 percent, at epoch 25 by replacing them with independent values in \([0,1]\). The paired \SpikeTrim and naive-mean runs receive the same initial states, channel samples, and corruption pattern.

The median correct range jumps from \(0.0163\) to \(0.8675\). \SpikeTrim reduces it to \(0.1444\) after one recovery update, \(0.0578\) after two, and \(0.0170\) after five. The naive mean immediately returns to its Byzantine-biased range near \(0.205\) and remains there. Figure~\ref{fig:recovery} illustrates the agreement-layer recovery guaranteed by Theorem~\ref{thm:self-stabilization}.

\begin{figure}[t]
  \centering
  \includegraphics[width=0.75\linewidth]{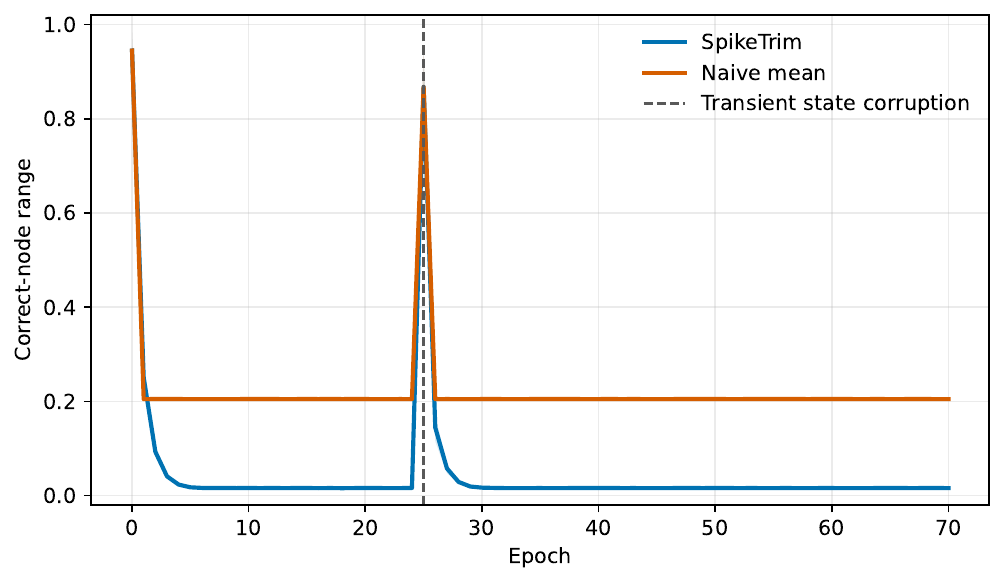}
  \caption{Recovery after transient corruption of 13 of the 32 correct scalar states at epoch 25. \SpikeTrim returns to its timing-limited regime. The experiment assumes that the \SIF timing, identity, and epoch services remain clean.}
  \label{fig:recovery}
\end{figure}
\section{Implementation implications and limitations}
\label{sec:discussion}

\subsection{Communication and local cost}

Each correct node emits one logical control event per epoch. Every correct receiver admits at most one event for each of the \(n\) sender labels. The correct broadcasts create \(O(n^2)\) logical sender-receiver observations. Efficient multicast can reduce physical link transmissions, but it does not change the number of receiver-side admissions. Byzantine flooding can create additional physical traffic, which must be controlled below \SIF as assumed in Section~\ref{sec:timing-model}.

A straightforward receiver stores \(n\) decoded values and sorts them in \(O(n\log n)\) time. The two trim thresholds can instead be found by linear-time selection, followed by one pass over the retained values. The memory cost remains \(O(n)\). Sender labels and epoch identifiers are metadata assumptions. On an address-event fabric that already transports source addresses, the scalar itself adds no payload bits, but the method is not an unlabeled bare-spike protocol.

The local sentinel does not require a transmitted event. Nominal-delay subtraction can use a calibrated per-link value or the midpoint of a known interval. Timestamp quantization contributes to \(\eta\). Equation~\eqref{eq:L-budget} then translates the measured residual uncertainty into a minimum usable timing window.

\subsection{Where the abstraction applies}

The strongest use case is a paced, sender-addressed neuromorphic fabric that already has a low-rate calibration or maintenance epoch. The state should be a slow scalar whose copies need mutual consistency. It need not be a membrane potential and the proof does not assume a leaky integrate-and-fire equation. The convex parameter \(\beta\) is an algorithmic relaxation.

If digital scalar packets are reliable and inexpensive, they avoid the timing precision floor. If the system is fully asynchronous, then a one-spike time coordinate lacks a shared origin and a different communication model is required. If source labels can be forged, one physical component can create Sybil values and the \(f\)-trimming proof fails.

\subsection{Limitations}

The theory assumes a fixed Byzantine set after the first clean epoch, complete logical visibility, a clean pacing and identity layer, and deterministic bounds on correct timing error. It does not provide physical denial-of-service protection before \FirstSpike filtering. It does not cover sparse graphs, vector-valued states, mobile faults, or time-varying membership. For \(\rho>0\), it provides one-step input-relatedness and asymptotic range agreement rather than global validity relative to the initial input hull. Remark~\ref{rem:cumulative-drift} gives the finite-horizon location bound and shows why periodic local re-estimation and calibration remain necessary.

The exact decoder lower bound applies to one deterministic observation in one epoch. Multi-spike and multi-epoch codes may improve precision by spending more time or events. The independent centered timing samples used in the simulations do not imply a stochastic location-drift theorem. Interval decoding, boundary clipping, order-statistic selection, and Byzantine values need not preserve conditional unbiasedness. A martingale-type result would require additional assumptions on decoded errors, the operating interval, and Byzantine behavior.

The analysis assumes exact real-valued storage and arithmetic for the agreement state. Finite-precision update error would introduce an additional perturbation in the range recurrence and must be accounted for in a hardware implementation.

The classifier experiment is a mechanism stress test on the \(8\times8\) scikit-learn digits data set, chosen to isolate the control-plane effect without the confounding influence of a large training pipeline or hardware-specific optimization. It is not a competitive SNN benchmark and does not establish generalization to MNIST, Fashion-MNIST, N-MNIST, or DVS-Gesture. No inferential comparison across independently trained model populations is reported. The experiments do not measure chip energy, router congestion, hardware mismatch distributions, or closed-loop hardware behavior.

\section{Conclusion}
\label{sec:conclusion}

A sender-labeled spike time is not automatically a Byzantine approximate-agreement message. Delay perturbs its value, silence removes an entry, flooding creates too many entries, and equivocation gives different receivers different timings. \SIF makes the required interface explicit. It combines paced epochs, sender attribution, per-label \FirstSpike admission, bounded timing uncertainty, and a sentinel for silence.

For this affine channel, midpoint feasible-set decoding gives the exact radius \(\rho=\min\{1/2,\omega/(T-2g)\}\). \SpikeTrim then applies the classical MSR trimmed mean to one effective value per label. For \(n\ge3f+1\), it has one-step robust validity and recovers the classical tight noiseless contraction factor under the direct update. It also has an explicit worst-case asymptotic range bound and a geometric recovery bound after agreement-layer state corruption once the interface assumptions hold again. The controlled task benchmark illustrates an association between faster alignment and lower cross-tile prediction disagreement under a finite maintenance budget. Near the Byzantine threshold, both convergence time and timing-noise amplification become severe. The closed-form criterion therefore tests whether the proved \SpikeTrim bound meets a required control-plane target. It does not replace platform measurement or establish an impossibility result for every one-spike protocol.

Our immediate next step is hardware validation. We plan to build the complete control path as a prototype on a field-programmable gate array (FPGA). The prototype will expose timestamping, epoch tagging, per-label \FirstSpike admission, sentinel insertion, decoding, and MSR as independently measurable modules. The subsequent step is deployment on sender-addressed neuromorphic hardware. That study will measure \(D\), \(\eta\), guard slack, label integrity, congestion sensitivity, endpoint memory, communication energy, and closed-loop calibration behavior. It will also extend the task study to a larger image benchmark and an event-based benchmark. These measurements will turn Corollary~\ref{cor:design-budget} from an illustrative calculation into a platform-specific engineering verdict. Other extensions include graph-robust MSR over sparse event networks, adaptive calibration of \(\bar d_{ij}\), and stochastic error bounds under explicit distributional assumptions.

\section*{Funding}

\noindent
The research of Arman Ferdowsi was funded by the Austrian Science Fund (FWF) \href{https://www.fwf.ac.at/en/research-radar/10.55776/ESP1705325}{10.55776/ESP1705325} (\href{https://ucrisportal.univie.ac.at/en/projects/symbolische-zeitanalyse-asynchroner-schaltungen/}{STAAC Project}).
The research of Atakan Aral was funded in part by CHIST-ERA-22-SPiDDS-07 (\href{https://troci.holisun.com/}{TROCI Project}) and by the Austrian Science Fund (FWF) \href{https://doi.org/10.55776/I6647}{10.55776/I6647}.

%

\section*{Declaration of competing interest}

\noindent
The authors declare that they have no known competing financial interests or personal relationships that could have appeared to influence the work reported in this paper.

%

\begingroup
\small
\bibliographystyle{elsarticle-num}
\bibliography{references}
\endgroup
\end{document}